\documentclass{IEEEtran}
\usepackage{enumerate}

\usepackage{pgfplots}
\usepackage{mathtools}
\usepackage{enumerate}
\usepackage{amsmath}
\usepackage{amssymb}
\usepackage{float}
\usepackage{xspace}
\usepackage{lettrine}
\usepackage{graphicx}
\usepackage{subfigure}
\usepackage{latexsym}
\usepackage{graphicx}
\usepackage{caption}
\usetikzlibrary{decorations.markings}
\usepackage{fancybox}
\usepackage{amsthm}
\newtheorem{theorem}{Theorem}
\newtheorem{definition}{Definition}[section]
\newtheorem{lemma}[theorem]{Lemma}

\newcommand{\R}{\mathbb{R}}

\IEEEoverridecommandlockouts

\begin{document}
\title{A New Theoretical Evaluation Framework for \\Satisfaction Equilibria in Wireless Networks}
\author{\IEEEauthorblockN{Michail Fasoulakis\IEEEauthorrefmark{1},
Eirini Eleni Tsiropoulou\IEEEauthorrefmark{2},
Symeon Papavassiliou\IEEEauthorrefmark{1}}\\
 \{mfasoul@netmode.ntua.gr, eirini@unm.edu, papavass@mail.ntua.gr\} \\
\IEEEauthorblockA{\IEEEauthorrefmark{1}School of Electrical and Computer Engineering,\\ National Technical University of Athens (NTUA), Greece}\\
\IEEEauthorblockA{\IEEEauthorrefmark{2}Department of Electrical and Computer Engineering,\\
 University of New Mexico, USA}}
\maketitle

\begin{abstract}
In this paper, a theoretical evaluation framework regarding the \textit{Satisfaction Equilibrium (SE)} in wireless communication networks is introduced and examined. To study these equilibria operation points, we coin some new concepts, namely the \textit{Valued Satisfaction Equilibrium}, the \textit{Price of Efficiency} and the \textit{Max Price of Satisfaction}, which can be used for measuring the efficiency of the obtained equilibria solutions. The aforementioned framework is analyzed and evaluated in a wireless communication environment under the presence of the Gaussian Interference channel (GIC). Within this setting, a non-cooperative game among the users is studied, where users aim in a selfish manner to meet their Quality of Service (QoS) prerequisite. However instead of maximizing the QoS which is generally energy costly, we evangelize that better energy-efficiency is achieved by targeting satisfactory QoS levels only. The sufficient and necessary conditions that lead to the \textit{Satisfaction Equilibrium} are provided for the two-user case and the \textit{Efficient Satisfaction Equilibrium (ESE)} is determined, where the users satisfy their QoS constraints with the lowest possible cost. Moreover, specific measures for evaluating the efficiency of various satisfaction equilibria, in a formal and quantitative manner, expressing the tradeoff with respect to the achieved utility or a given objective function and corresponding cost, are defined and analyzed.  
\end{abstract}


\section{Introduction}
Current and future wireless networks face new and interesting challenges in the increasingly changing communications environment. The volume of the transmitted data traffic and the number of users are continuously increasing. Typically, these systems present competitive environments that induce constraints, and users evolve with others, where their decisions and actions are interdependent. 

In such a competitive and distributed environment, \textit{Game Theory} arises as a natural choice and a powerful theoretical tool to cope with the corresponding resource allocation problems, while reflecting and modeling the different interactions among the users \cite{lasaulce2011game}. Respecting the need for distributed and scalable solutions and algorithms, the focus has been placed on the study of non-cooperative game theory paradigm where decisions are taken autonomously by the end users \cite{elhammouti2017self}. 

The majority of existing approaches have relied on the principles of Expected Utility Maximization, where users aim at selfishly maximizing their own degree of satisfaction upon receiving a service, as expressed through various forms of utility functions, resulting to some equilibrium \cite{tsiropoulou2017supermodular}. However, firstly a major disadvantage that is common to most of these equilibrium concepts (e.g., Nash equilibrium) is that stability depends on whether or not each user achieves the highest performance possible. This does not necessarily properly reflect in reality the most desirable solution, neither from a single user point of view nor from a network point of view, as a user may require only to ensure a specific minimum QoS condition \cite{southwell2014quality}. To overcome this constraint and problem, in this work a new solution concept known as Satisfaction Equilibrium (SE) \cite{ross2006satisfaction} is adopted and exploited in the realm of wireless communications environments. The main purpose of this paper is to introduce a holistic theoretical analysis and evaluation framework for studying these equilibria in wireless communication networks, with respect to conditions of existence and their efficiency. Its application in the presence of \textit{Gaussian Interference channel (GIC)} is specifically analyzed and examined.  

\subsection{Related work and our contribution}
Most of the efforts so far in the wireless communications research community were concentrated on the total QoS maximization, which however resulted in unjustified network resources drainage or unfair resources allocation among the users \cite{tsiropoulou2016uplink}. Instead, in our work to overcome these deficiencies, we adopt the philosophy of aiming satisfactory QoS \cite{ross2006satisfaction,meriaux2012achievability,perlaza2012quality} rather than targeting optimality in terms of QoS maximization \cite{katsinis2017joint}. There are two key motivating factors behind this consideration. The first one refers to meeting user expectations. Several types of services are either simply interested in achieving a minimum QoS level, or corresponding users are insensitive to small QoS changes \cite{ross2006satisfaction}. Furthermore, users may not be willing to consume additional resources or pay higher price for a better QoS level, which in turn translates to an increase in network capacity in terms of satisfied users \cite{meriaux2012achievability}. The second factor behind the choice of satisfactory solutions stems from the need of energy consumption reduction \cite{perlaza2012quality, goonewardena2017generalized}. QoS maximization typically requires unnecessarily high energies \cite{tsiropoulou2017supermodular}, whereas achieving an efficient and feasible network operation point that meets users’ expectations can result in substantial energy savings, a critical factor in several resource-constrained environments. 

This objective in this paper is treated via a novel mathematical concept within the framework of game theory, referred to as satisfaction equilibrium (SE) \cite{perlaza2012quality}. For studying these equilibria, we introduce and coin some new concepts, namely the \textit{Valued Satisfaction equilibrium}, the \textit{Price of Efficiency} and the \textit{Max Price of Satisfaction}, that can be used for benchmarking purposes of the achievable solutions. The applicability of the aforementioned framework and concepts are demonstrated considering a wireless communication environment characterized by Gaussian Interference channel, where the users have limited transmission power \cite{fasoulakis2017gaussian}. Towards this direction, we formulated the \textit{Gaussian Interference channel} as a non-cooperative game among the users, where each user has a minimum QoS prerequisite, represented as a minimum achievable data rate. Specific necessary and sufficient conditions that lead to the SE for the two-user case are obtained. The satisfaction equilibrium point is achieved when all the users of the network (i.e., players in the game) satisfy their minimum QoS requirements irrespective of the utility value they achieve. Relaxing the maximization assumptions essentially we enlarge the potential set of feasible strategies since instead of restricting ourselves to solutions that permit global optimum, we follow a less restrictive approach and extend the solution space to a broader set. Subsequently the main objective becomes - in contrast to reaching other forms of equilibria such as Nash equilibrium (NE) - to reach a network state in which all end users satisfy their individual QoS by investing the minimum effort, referred to as Efficient Satisfaction Equilibrium (ESE) \cite{meriaux2012achievability}. 

\subsection{Outline}

The remaining of this work is organized as follows. In Section \ref{sec:Game Theory}, the fundamental concepts from the field of Game Theory are initially presented. Therefore, concepts such as satisfaction equilibrium and efficient satisfaction equilibrium are reviewed, while the new concepts and terms of valued satisfaction equilibrium, price of efficiency, and max price of satisfaction are defined. In Section \ref{sec:Model}, the baseline proposed Gaussian Interference Channel model is introduced, as well as the users' adopted utility function and accordingly the corresponding game theoretic formulation of satisfying users' QoS prerequisites is presented. Section \ref{sec:Existence} discusses the conditions that lead to the satisfaction equilibrium for the two-user case, while in Section \ref{sec:PoE} the price of efficiency and max price of satisfaction for evaluating the efficiency of the various SEs in a formal and quantitative manner, are analyzed. Finally Section \ref{sec:Conclusions} concludes this article. 

\section{Game theory definitions} \label{sec:Game Theory}
We consider formally a game in satisfaction form as a tetrad 
\[G = \Big(N, \{A_i\}_{i \in N}, \{u_i\}_{i\in N}, \{f_i\}_{i\in N}\Big),\]
where $N$ is the set of the players with size $|N|$, $A_i$ is the set of all possible pure strategies (actions) of the player $i$, $u_i$ is the utility function of the player $i$ and $f_i$ is the set of all satisfied actions under a constraint of the player $i$ given the actions of all other players.
A strategy profile is an $|N|$-tuple $\mathbf{a}=(a_1,a_2, \dots, a_{|N|}) \in A_1\times A_2\times \dots \times A_{|N|}$ of the strategies of the players. Also, $\mathbf{a_{-i}}$ is the strategy profile of all players except for the player $i$.
We have the following definitions.

\begin{definition}[Nash equilibrium \cite{nash1951non}]
A Nash equilibrium is a strategy profile $\mathbf{a^*} = (a_i^*, \mathbf{a_{-i}^*})$ such that, for any $i$ and any strategy $a_i \in A_i$,
\[u_i(a_i^*,\mathbf{a_{-i}^*}) \geq u_i(a_i,\mathbf{a_{-i}^*}).\]
\end{definition}

\begin{definition}[Generalized Nash equilibrium \cite{debreu1952social}] 
A Generalized Nash equilibrium is a strategy profile $\mathbf{a^*} = (a_i^*, \mathbf{a_{-i}^*})$ such that, for any $i$,
\[a_i^* \in f_i(\mathbf{a_{-i}^*}), \text{ and}\]
\[u_i(a_i^*,\mathbf{a_{-i}^*}) \geq u_i(a_i,\mathbf{a_{-i}^*}),\]
for any $a_i \in f_i(\mathbf{a_{-i}^*})$.
\end{definition}

\begin{definition}[Satisfaction equilibrium \cite{perlaza2012quality}]
A satisfaction equilibrium is a strategy profile $\mathbf{a^*} = (a_i^*, \mathbf{a_{-i}^*})$ such that, for any $i$,
\[a_i^* \in f_i(\mathbf{a_{-i}^*}).\]
\end{definition}

\begin{definition}[Efficient satisfaction equilibrium \cite{perlaza2012quality}]
An efficient satisfaction equilibrium is a strategy profile $\mathbf{a^*} = (a_i^*, \mathbf{a_{-i}^*})$ such that, for any $i$,
\[a_i^* \in f_i(\mathbf{a_{-i}^*}), \text{ and }c_i(\mathbf{a^*}) \leq c_i(\mathbf{a}),\]
for any $\mathbf{a} = (a_i, \mathbf{a_{-i}^*})$ with $a_i \in f_i(\mathbf{a_{-i}^*})$, where $c_i: A_1 \times A_2 \times \dots A_{|N|} \mapsto \R_{+}$ is a cost function.
\end{definition}

In the following we first introduce the \textit{Valued Satisfaction equilibrium}, which essentially corresponds to efficient satisfaction equilibria, with a new cost function ${c_i}/{u_i}$ measuring the tradeoff between the achieved utility $u_i$ and the corresponding initial cost $c_i$. Specifically:

\begin{definition}[Valued satisfaction equilibrium]
A valued satisfaction equilibrium is a strategy profile $\mathbf{a^*} = (a_i^*,\mathbf{a_{-i}^*})$ such that, for any $i$,
\[a_i^* \in f_i(\mathbf{a_{-i}^*}), \text{ and}\]
\[\frac{c_i(\mathbf{a^*})}{u_i(\mathbf{a^*})} \leq \frac{c_i(\mathbf{a})}{u_i(\mathbf{a})},\]
for any for any $\mathbf{a} = (a_i, \mathbf{a_{-i}^*})$ with $a_i \in f_i(\mathbf{a_{-i}^*})$, where $c_i$ is the cost function and $u_i$ the utility function, assuming that $c_i(a)>0$ and $u_i(a)>0$ for any $a \in A_1 \times A_2 \times \dots A_{|N|}$. 
\end{definition}

Below, we also introduce a new concept, referred to as Price of Efficiency, to measure the efficiency of the efficient satisfaction equilibria. Intuitively, this concept aims to evaluate in a formal and quantitative manner, the efficient satisfaction equilibria solution, in terms of tradeoff between the achieved utility and cost, when compared to the valued satisfaction equilibrium operation point. 

\begin{definition}[Price of Efficiency]
Let $\mathbf{a^*}$ be a strategy profile such that it is the best valued satisfaction equilibrium (under the summation function). Furthermore, let $\mathbf{a}$ be a strategy profile such that $\mathbf{a}$ is the worst efficient satisfaction equilibrium (under the summation function). Then, the Price of Efficiency (PoE) is
\[PoE = \frac{\sum\limits_{\forall i}\frac{c_i(\mathbf{a})}{u_i(\mathbf{a})}}{\sum\limits_{\forall i}\frac{c_i(\mathbf{a^*})}{u_i(\mathbf{a^*})}},\]
assuming that $c_i(a)>0$ and $u_i(a)>0$ for any $a \in A_1\times A_2\times \dots \times A_{|N|}$.
\end{definition}

Last but not least, we also define a measure indicating the maximum price (i.e., distance) characterizing the achievable satisfaction equilibria, under the consideration of an objective function $g$. Formally: 
\begin{definition}[Max Price of Satisfaction]
Let $\mathbf{a^*}$ be a strategy profile such that it is a satisfaction equilibrium and it is an optimal of an objective function $g$, i.e., $g$ is the sum of the utilities (utilitarian optimum). Furthermore, let $\mathbf{a}$ be a strategy profile such that $\mathbf{a}$ is the worst satisfaction equilibrium under the objective function $g$. Then, the Max Price of Satisfaction (MPoSa) is
\[ MPoSa = \frac{g(\mathbf{a^*})}{g(\mathbf{a})},\]
assuming that $g(a)>0$ for any $a \in A_1\times A_2\times \dots \times A_{|N|}$.
\end{definition}
Since, $g(\mathbf{a^*}) \geq g(\mathbf{a})$, we have that the MPoSA is no less than one.
\section{The Model} \label{sec:Model}
We consider $|N|$ pairs of transmitter-receiver in the Gaussian Interference Channel. Any transmitter $i\in\{1,2,\dots,|N|\}$ uses the same frequency to transmit her message to her receiver $i$ and simultaneously interferes the receivers of the other transmitters. Any receiver $i$ treats the signal of a transmitter $j$ different than $i$ as interference. We can see this competitive situation as a non-cooperative game \cite{osborne1994course}. Any transmitter wants to maximize her rate.
The set of the pure strategies of any transmitter $i$ is the power $p_i \in [0, P_{\max}]$ that chooses to transmit her message, where $P_{\max}$ is the maximum possible power, imposed by physical constraints and equipment characteristics. The utility function of any transmitter $i$ is

\begin{equation}
\label{utility function}
u_i(p_i,\mathbf{p_{-i}}) = \frac{1}{2}\log\Bigg(1+\frac{{h_{ii}p_i}}{{\sum\limits_{\substack{j=1 \\ j\neq i}}^{|N|}h_{ji}p_j+ I}}\Bigg),
\end{equation}
where $h_{ji}>0$ is the channel attenuation between the transmitter $j$ and the receiver $i$, $\mathbf{p_{-i}}$ is the vector of the powers of the transmitters except for the transmitter $i$ and $I>0$ is the \emph{Additive White Gaussian Noise (AWGN)} of the channel. This can be written as 
\begin{equation}
\label{utility function2}
u_i(p_i,\mathbf{p_{-i}}) = \frac{1}{2}\log\Bigg(1+\frac{{p_i}}{{\sum\limits_{\substack{j=1 \\ j\neq i}}^{|N|}a_{ji}p_j+ I_i}}\Bigg),
\end{equation}
where $a_{ji} = h_{ji}/h_{ii}>0$ is the new channel attenuation between the transmitter $j$ and the receiver $i$ and $I_i = I/h_{ii}$. We give the following definitions in the context of the Gaussian Interference channel.

\begin{definition}[Nash equilibrium]
A strategy profile $(p_i^*, \mathbf{p_{-i}^*})$ is a Nash equilibrium such that, for any $i$ and for any $p_i$,
\[u_i(p_i^*,\mathbf{p_{-i}^*}) \geq u_i(p_i,\mathbf{p_{-i}^*}),\]
in other words $p_i^*$ is a best-response strategy to the strategies $\mathbf{p_{-i}^*}$ of the other transmitters.
\end{definition}

\begin{definition}[Generalized Nash equilibrium]
A strategy profile $(p_i^*, \mathbf{p_{-i}^*})$ is a generalized Nash equilibrium such that, for any $i$ and for any $p_i$ such that $u_i(p_i,\mathbf{p_{-i}^*}) \geq \Gamma_i$, we have
\[u_i(p_i^*,\mathbf{p_{-i}^*}) \geq u_i(p_i,\mathbf{p_{-i}^*}),\]
where $\Gamma_i$ is a positive constant representing user's $i$ QoS prerequisites in terms of achievable transmission data rate.
\end{definition}

\begin{definition}[Satisfaction equilibrium]
A strategy profile $(p_i^*, \mathbf{p_{-i}^*})$ is a satisfaction equilibrium such that, for any $i$,
\[u_i(p_i^*,\mathbf{p_{-i}^*}) \geq \Gamma_i,\]
in other words $p_i^*$ is a satisfaction strategy to the strategies $\mathbf{p_{-i}^*}$ of the other transmitters, $\Gamma_i$ is the satisfaction threshold (positive constant) of the transmitter $i$.
\end{definition}

\begin{definition}[Efficient satisfaction equilibrium]
A strategy profile $(p_i^*, \mathbf{p_{-i}^*})$ is an efficient satisfaction equilibrium, if it is a satisfaction equilibrium and, for any $i$ and any $p_i$ such that $u_i(p_i,\mathbf{p_{-i}^*}) \geq \Gamma_i$, we have
\[p_i^* \leq p_i,\]
or in other words the transmitter $i$ uses the less power to be satisfied.
\end{definition}

\begin{definition}[Valued satisfaction equilibrium]
A strategy profile $(p_i^*, \mathbf{p_{-i}^*})$ is a valued satisfaction equilibrium, if it is a satisfaction equilibrium and, for any $i$ and any $p_i$ such that $u_i(p_i,\mathbf{p_{-i}^*}) \geq \Gamma_i$, we have
\[\frac{p_i^*}{u_i(p_i^*,\mathbf{p_{-i}^*})} \leq \frac{p_i}{u_i(p_i,\mathbf{p_{-i}^*})}.\]

\end{definition}

\section{Existence of SE} \label{sec:Existence}

The satisfaction equilibria of the Gaussian Interference channel have been studied previously in \cite{meriaux2012achievability,perlaza2012quality} but no conditions of existence have been given. In this paper we will investigate the conditions of existence of satisfaction equilibria in this scenario.
By the utility function of the player $i$ we have
\begin{align}
\label{area of SE}
u_i(p_i,\mathbf{p_{-i}}) =  \frac{1}{2}\log\Bigg(1+\frac{{p_i}}{{\sum\limits_{\substack{j=1 \\ j\neq i}}^{|N|}a_{ji}p_j+ I_i}}\Bigg) \geq \Gamma_i \\
\Leftrightarrow p_{i} \geq (4^{\Gamma_i}-1)\Big(\sum\limits_{\substack{j=1 \\ j\neq i}}^{|N|}a_{ji}p_j+I_i\Big).
\end{align}
These are the inequalities that characterize the region of the SE. We can write up the inequalities (\ref{area of SE}) as a linear system 
\begin{align}
\label{linear system}
\pmb{A}\cdot \pmb{p} \geq \pmb{b},
\end{align}
where \\
\begin{multline*}
\pmb{A} = \begin{pmatrix} 1 && a_{21}(1-4^{\Gamma_1}) && \dots && a_{|N|1}(1-4^{\Gamma_1}) \\ a_{12}(1-4^{\Gamma_2}) && 1  && \dots && a_{|N|2}(1-4^{\Gamma_2}) \\ &&  && \dots && && \\ a_{1|N|}(1-4^{\Gamma_{|N|}}) && a_{2|N|}(1-4^{\Gamma_{|N|}}) && \dots &&  1 \end{pmatrix},\\\\
\pmb{p} = \begin{pmatrix} p_1 \\ p_2 \\ \dots \\ p_{|N|} \end{pmatrix}, \text{ and }
\pmb{b} = \begin{pmatrix} (4^{\Gamma_1}-1)I_1 \\ (4^{\Gamma_2}-1)I_2 \\ \dots \\ (4^{\Gamma_{|N|}}-1)I_{|N|}
\end{pmatrix}.
\end{multline*}
There is at least one SE if there is at least one feasible solution in the system $\pmb{A}\cdot \pmb{p} \geq \pmb{b}$ .

\subsection{Two-player game}
We study the conditions of existence of the SE in the case of the two players and we conclude to the following Theorems.
\begin{theorem}
The satisfaction equilibrium region of the two-player game is the region such that $p_1 \geq (4^{\Gamma_1}-1)(a_{21}p_2+I_1)$ and $p_2 \geq (4^{\Gamma_2}-1)(a_{12}p_1+I_2)$.
\end{theorem}
\begin{proof}
By the utility function of the player 1 we have
\begin{multline*}
u_1(p_1,p_2) = \frac{1}{2}\log\Big(1+\frac{p_1}{a_{21}p_2+I_1}\Big)\geq \Gamma_1\\ \Leftrightarrow p_1 \geq (4^{\Gamma_1}-1)(a_{21}p_2+I_1).
\end{multline*}
Similarly, for the player 2 we have
\begin{multline*}
u_2(p_2,p_1) = \frac{1}{2}\log\Big(1+\frac{p_2}{a_{12}p_1+I_2}\Big)\geq \Gamma_2\\ \Leftrightarrow p_2 \geq (4^{\Gamma_2}-1)(a_{12}p_1+I_2).
\end{multline*}
Since, $\Gamma_1$ and $\Gamma_2$ are positive is clear that $(4^{\Gamma_1}-1)(a_{21}p_2+I_1) > 0$ and $(4^{\Gamma_2}-1)(a_{12}p_1+I_2) > 0$.
\end{proof}

\begin{theorem}
\label{intersection}
There is at least one satisfaction equilibrium if and only if the lines \begin{equation} \label{line 1} p_2 = {p_1}/({(4^{\Gamma_1}-1)a_{21}})-{I_1}/{a_{21}}\end{equation} and \begin{equation}\label{line 2} p_2 = (4^{\Gamma_2}-1)(a_{12}p_1+I_2) \end{equation} cross in the first quadrant, or in other words $a_{21}a_{12} < \frac{1}{(4^{\Gamma_1}-1)(4^{\Gamma_2}-1)}$.
\end{theorem}
\begin{proof}
Graphically, there is at least one feasible solution in the area described by $p_1 \geq (4^{\Gamma_1}-1)(a_{21}p_2+I_1)$ and $p_2 \geq (4^{\Gamma_2}-1)(a_{12}p_1+I_2)$ if and only if the two lines $p_2 = \frac{p_1}{(4^{\Gamma_1}-1)a_{21}}-\frac{I_1}{a_{21}}$ and $p_2 = (4^{\Gamma_2}-1)(a_{12}p_1+I_2)$ cross each other in the first quadrant. In order the two lines to cross each other in the first quadrant the slope of the first one must be greater than the slope of the second one. Thus, $\frac{1}{(4^{\Gamma_1}-1)a_{21}} > (4^{\Gamma_2}-1)a_{12}$.
\end{proof}
\begin{figure}
\centering
\captionsetup{justification=centering}
\subfigure{
\begin{tikzpicture}[scale = 1.5]
\draw [<->,thick] (0,2) node (yaxis) [above] {$p_2$}
     |- (2.7,0) node (xaxis) [right] {$p_1$};
\draw (0.5,0) coordinate (a_1) -- (2,1.7) coordinate (a_2);
\draw (0.5,0) node [below] {$(4^{\Gamma_1}-1)I_1$};
\draw (0,0.5) coordinate (b_1) -- (2.5,1) coordinate (b_2);
\draw (0,0.5) node [left] {$(4^{\Gamma_2}-1)I_2$};
\coordinate (c) at (intersection of a_1--a_2 and b_1--b_2);
\draw (c) node [above left] {ESE};
\fill[red] (c) circle (1pt);
\draw (1.3,0.9) -- (1.3,0.76);
\draw (1.4,1.02) -- (1.4,0.78);
\draw (1.5,1.13) -- (1.5,0.8);
\draw (1.6,1.24) -- (1.6,0.82);
\draw (1.7,1.36) -- (1.7,0.84);
\draw (1.8,1.47) -- (1.8,0.86);
\draw (1.6,1.2) node [above left] {$(\ref{line 1})$};
\draw (2,0.4) node [above left] {$(\ref{line 2})$};
\end{tikzpicture}
}
\caption{The area of the satisfaction equilibria}
\end{figure}
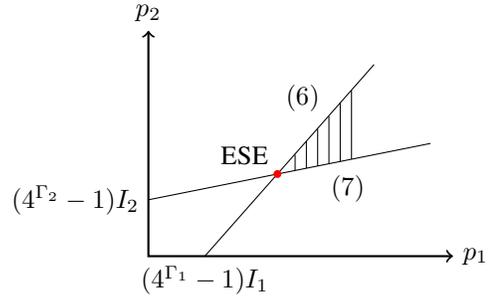

\begin{theorem}
The efficient satisfaction equilibrium exists if and only if $a_{21}a_{12} < \frac{1}{(4^{\Gamma_1}-1)(4^{\Gamma_2}-1)}$ and is equal to $\Big(\frac{(4^{\Gamma_1}-1)(a_{21}(4^{\Gamma_2}-1)I_2+I_1)}{1-(4^{\Gamma_1}-1)(4^{\Gamma_2}-1)a_{21}a_{12}},\frac{(4^{\Gamma_2}-1)(a_{12}(4^{\Gamma_1}-1)I_1+I_2)}{1-(4^{\Gamma_1}-1)(4^{\Gamma_2}-1)a_{21}a_{12}}\Big)$.
\end{theorem}
\begin{proof}
For fixed strategy of the one player, the other player tries to be satisfied with the minimum power in the satisfaction equilibrium region. So, it is easy to see that the efficient satisfaction equilibrium must be on all boundaries of the satisfaction region. This is the cross point of the lines (\ref{line 1}) and (\ref{line 2}).
The cross point is the solution of the system
\begin{align*}
p_1 = (4^{\Gamma_1}-1)(a_{21}p_2+I_1),\\
p_2 = (4^{\Gamma_2}-1)(a_{12}p_1+I_2),
\end{align*}
which has solution the pair $\Big(\frac{(4^{\Gamma_1}-1)(a_{21}(4^{\Gamma_2}-1)I_2+I_1)}{1-(4^{\Gamma_1}-1)(4^{\Gamma_2}-1)a_{21}a_{12}},\frac{(4^{\Gamma_2}-1)(a_{12}(4^{\Gamma_1}-1)I_1+I_2)}{1-(4^{\Gamma_1}-1)(4^{\Gamma_2}-1)a_{21}a_{12}}\Big)$.
We want this solution to be non negative. Since the nominators are non negative the denominators must be positive, so $1-(4^{\Gamma_1}-1)(4^{\Gamma_2}-1)a_{21}a_{12}>0$. But this holds since $a_{21}a_{12} < \frac{1}{(4^{\Gamma_1}-1)(4^{\Gamma_2}-1)}$.
\end{proof}
By this theorem we conclude to the following lemma.
\begin{lemma}
The efficient satisfaction equilibrium is unique.
\end{lemma}

We now search for the valued satisfaction equilibria of the game and we conclude to the following theorem.
\begin{theorem}
\label{Valued SE}
The valued satisfaction equilibrium coincides with the efficient satisfaction equilibrium.
\end{theorem}
\begin{proof}
We do the analysis from the player 1 point of view, the analysis for the other player is symmetric.
Fix a strategy $p_2^*$ of the player 2. Then, we search for the strategy of the player 1 that minimizes the $\frac{p_1}{\frac{1}{2}\log(1+\frac{p_1}{a_{21}p_2^*+I_1})}$. The derivative of this is
\begin{multline*}
\frac{\frac{1}{2}\log(1+\frac{p_1}{a_{21}p_2^*+I_1}) - \frac{\frac{1}{2}p_1}{p_1+a_{21}p_2^*+I_1}}{\frac{1}{4}(\log(1+\frac{p_1}{a_{21}p_2^*+I_1}))^2}\\ 
= \frac{\frac{1}{2}\log(1+\frac{p_1}{a_{21}p_2^*+I_1}) - \frac{1/2}{1+\frac{a_{21}p_2^*+I_1}{p_1}}}{\frac{1}{4}(\log(1+\frac{p_1}{a_{21}p_2^*+I_1}))^2},
\end{multline*}
which is positive for $p_1>0$, so in this interval the function is increasing. Thus, the best strategy for the player 1 is the minimum in the satisfaction equilibrium area. This means similarly to the ESE that the valued satisfaction equilibrium is on all boundaries of the satisfaction equilibrium area, so on the cross point.
\end{proof}

\section{The Price of Efficiency and the Max Price of Satisfaction} \label{sec:PoE}
We study the Price of Efficiency in the Gaussian Interference channel as a measurement of the efficiency of the efficient satisfaction equilibrium in terms of the tradeoff between the achieved rate and the power cost to achieve this rate. We conclude to the following theorem.

\begin{theorem}
The Price of Efficiency in the Gaussian Interference channel is equal to 1.
\end{theorem}
\begin{proof}
Since, the efficient satisfaction equilibrium coincides with the valued satisfaction equilibrium it is easy to see that the PoE is equal to 1.
\end{proof}

We now study the maximum price of satisfaction taking as an objective function a function that minimizes the sum of the powers, or in other words maximizes the ratio one over the sum of the powers. It is easy to see that the efficient satisfaction equilibrium is the optimum solution of this objective function in the area of the satisfaction equilibria. On the other hand assuming that the $(P_{\max},P_{\max})$ is in the satisfaction equilibrium area, this is the worst case of this objective function. Thus, it is easy to see that the maximum price of satisfaction is equal to the ratio $\frac{2P_{\max}}{\sum\limits_{\forall i} P_i}$, where $P_i$ is the power of the player $i$ in the ESE. 

\section{Concluding Remarks} \label{sec:Conclusions}
In this paper, we argued that by rethinking the overall traditional QoS provisioning and user experience perception in wireless networks that targeted QoS maximization, we can devise energy-efficient, scalable, and rewarding solutions from both practical and theoretical viewpoints, through the adoption of the general concept of satisfaction equilibrium. Therefore, a novel holistic framework was introduced towards studying different efficient and feasible wireless network operation points (i.e., different game theoretic satisfaction equilibria points) that meet users’ expectations, while resulting in substantial cost or energy savings, which are critical factors especially in resource-constrained environments. This new framework was in particular analyzed in a wireless communication environment under the presence of the Gaussian Interference channel, while the efficiency of the various satisfaction equilibria was evaluated in a formal and quantitative manner. The current analysis can be extended via considering multi-service wireless communication environments, where the users' utilities are differentiated based on their requested services, i.e., real-time and non-real time services \cite{tsiropoulou2013efficient}. In this case, the study of Price of Efficiency is interesting and challenging, as the problem becomes more complicated indicating that the PoE may not be equal to $1$.

It should be noted that the realization of new decision-making paradigms addressing the requirements of flexibility, adaptability, and autonomicity in future wireless networks, depends on bridging the gap between the considered game theoretic methodologies described above, and practical wireless applications and settings. The efficient computation of the satisfaction equilibria is hard in practice even for games which satisfy the necessary conditions of existence and convergence to them in theory, due to the inherent difficulties that wireless networks face in practice. Part of our current and future research work, includes the use of learning approaches in game theory, in order to deal with the technical and implementation challenges stemming from the incompleteness of available information regarding the game structure and uncertainty on the observations of the users and their actions in the game that in turn may influence decision making and equilibrium identification and convergence. 

\section*{Acknowledgment}
\noindent
The research of Dr. Michail Fasoulakis and Dr. Symeon Papavassiliou was partially supported by the NTUA-GSRT Research Award under Grant Number 67104700. The research of Dr. Eirini Eleni Tsiropoulou was conducted as part of the UNM Research Allocation Committee award and the UNM Women in STEM Faculty Development Fund. 
\bibliographystyle{IEEEtran}
\bibliography{bibl}
\end{document}